\documentclass[12pt]{article}

\usepackage[english]{babel}
\usepackage[utf8x]{inputenc}
\usepackage[T1]{fontenc}
\usepackage{authblk}

\usepackage{latexsym}		
\usepackage{epsfig}		
\usepackage{graphicx}           
\usepackage{subcaption}

\usepackage{amsmath}	
\usepackage{bm}
\usepackage{amssymb}	
\usepackage{amsfonts}	
\usepackage{amsthm}	

\newtheorem{thm}{Theorem}[section]
\newtheorem{lem}[thm]{Lemma}
\newtheorem{cor}[thm]{Corollary}

\theoremstyle{definition}
\newtheorem{defn}[thm]{Definition}

\usepackage{latexsym,array,delarray,amsthm,amssymb,epsfig,color}
\usepackage{graphics,graphicx, caption, comment}
\usepackage{longtable,tikz}
\usepackage{float}

\usetikzlibrary{calc,intersections,through,backgrounds}

\pgfdeclarelayer{background}
\pgfsetlayers{background,main}

\theoremstyle{plain}

\newtheorem{prop}[thm]{Proposition}
\newtheorem*{thm*}{Theorem \ref{thm:main}}
\newtheorem*{lemma*}{Lemma}
\newtheorem*{prop*}{Proposition}
\newtheorem*{cor*}{Corollary}
\newtheorem*{conj*}{Conjecture}

\theoremstyle{definition}
\newtheorem{ex}[thm]{Example}

\theoremstyle{remark}

\newcommand{\R}{\mathbb{R}}

\newcommand{\C}{\mathcal C}

\newcommand{\ind}{\mbox{$\perp \kern-5.5pt \perp$}}

\DeclareMathOperator*{\supp}{supp}

\tikzstyle{vertex} = [fill,shape=circle,node distance=80pt]
\tikzstyle{edge} = [fill,opacity=.5,fill opacity=.5,line cap=round, line join=round, line width=50pt]
\tikzstyle{elabel} =  [fill,shape=circle,node distance=30pt]

\usepackage[a4paper,top=3cm,bottom=2cm,left=3cm,right=3cm,marginparwidth=1.75cm,margin=1in]{geometry}

\usepackage{amsmath}
\usepackage{graphicx}
\usepackage[colorinlistoftodos]{todonotes}
\usepackage[colorlinks=true, allcolors=blue]{hyperref}

\title{On the identification of $k$-inductively pierced codes using toric ideals}
\author[1]{Molly Hoch}
\author[2]{Samuel Muthiah}
\author[3]{Nida Obatake}
\affil[1]{Wellesley College}
\affil[2]{Westmont College}
\affil[3]{Texas A\&M University}
\date{July 6, 2018}

\begin{document}
\maketitle

\begin{abstract}
Neural codes are binary codes in $\{0,1\}^n$; here we focus on the ones which represent the firing patterns of a type of neurons called place cells. There is much interest in determining which neural codes can be realized by a collection of convex sets. However, drawing representations of these convex sets, particularly as the number of neurons in a code increases, can be very difficult. Nevertheless, for a class of codes that are said to be $k$-inductively pierced for $k=0,1,2$ there is an algorithm for drawing Euler diagrams. Here we use the toric ideal of a code to show sufficient conditions for a code to be 1- or 2-inductively pierced, so that we may use the existing algorithm to draw realizations of such codes. \end{abstract}

\section{Introduction}

John O'Keefe's discovery of place cells earned him a share of the 2014 Nobel Prize for Physiology or Medicine \cite{OD71}. Place cells are a type of neuron found in certain mammals that help them to locate themselves spatially. Place cells fire only when the mammal is in a certain part of its environment. The receptive fields, or areas where the neurons fire, are approximately convex, and the firing activity of place cells can be represented by neural codes. Much study has been done on convex neural codes, specifically focusing on which codes have convex receptive fields corresponding to them. These codes are called convex realizable \cite{cruz2016open,what-makes,FH-REU17,FM-REU17,lienkaemper2017obstructions,mulas2017characterization,rosen2017convex}. 

However, even if a convex realization exists, it can be difficult to draw one. Relative to the study of convex realizable codes, far less work has been done on how to draw realizations of neural codes. However, \cite{GOY16} have shown that the algorithm created by \cite{Stapleton10} to draw Euler diagrams can be used to create 2-dimensional realizations of codes that are 0-, 1-, or 2-inductively pierced.  In \cite{GOY16}, the authors gave a necessary and sufficient condition for identifying 0-inductively pierced codes, and a necessary condition for 1-inductively pierced codes using degree bounds on generating sets of a binomial ideal called the neural toric ideal.  They conjectured that degree bounds on generators of toric ideals could be used to determine sufficient conditions for 1- and 2-inductively pierced codes.  In this paper, we explore sufficient conditions for 1- and 2-inductively pierced codes by examining elements of the neural toric ideal.  We also introduce new algebraic signatures for 1- and 2-inductively pierced codes.  Specifically, Proposition~\ref{1-inductively_pierced_sufficient_condition} gives a sufficient condition for 1-inductively pierced codes, and Theorem~\ref{cubic_implies_2-piercing} gives a sufficient condition for 2-inductively pierced codes. 

In Section 2, we discuss the definitions necessary for our main results, which we present in Section 3.  We conclude with a discussion of our findings in Section 4.

\section{Background}
We follow the definitions from \cite{GOY16}.

\begin{defn}
A \textit{neural code} on $n$ neurons is a set of binary strings $\mathcal{C} \subseteq \{0,1\}^n$.  An element $\sigma$ of $\C$ is a \textit{codeword}.
\end{defn}

We will assume that the all-zeros codeword is always in a neural code.  For simplicity of notation, we will write $[n]:=\{1,2,\ldots,n\}$.
\begin{defn}
A \textit{realization} of a code $\mathcal{C}$ is a collection of sets $\mathcal{U}=\{U_1,\ldots,U_n\}$ where $U_i\subseteq \mathbb{R}^d$ such that $\mathcal{C}=\mathcal{C}(\mathcal{U}):=\{\sigma \subseteq [n] \mid U_\sigma \setminus \bigcup_{j\in[n]\setminus{\sigma}}U_j \neq \emptyset \}$ where for $\sigma \subseteq [n]$, we define $U_\sigma:=\bigcap_{i\in\sigma}U_i$.  A $U_i\in \mathcal{U}$ is the \textit{place field} of the neuron $i$.  A \textit{zone} in $\mathcal{C}$ is a region $U_\sigma \setminus \bigcup_{j\in[n]\setminus{\sigma}}U_j$ in a realization of $\C$ and each zone corresponds to a codeword.  
\end{defn}

\begin{defn}
A code $\C$ is \textit{convex open} if it is realizable by $\mathcal{U}=\{U_1,\ldots,U_n\}$ where all the $U_i$ are convex open sets. 
\end{defn}

In this paper, we focus on neural codes that are convex open and realizable in dimension 2.  Furthermore, we will assume that the boundary curves of place fields exist and that our codes are well-formed.

\begin{defn}
A code $\C$ is \textit{well-formed} if there exists a realization of $\C$ in $\R^2$ such that \begin{itemize}
\item The boundary curves of place fields intersect generally (i.e., not tangentially nor identically), and at only a finite number of points.
\item
At any given point, at most two boundaries of place fields intersect.
\item
Each zone is connected.
\end{itemize}

\end{defn}

Figure~\ref{wellformedviolationsfig} illustrates examples of realizations of codes that are not well-formed.
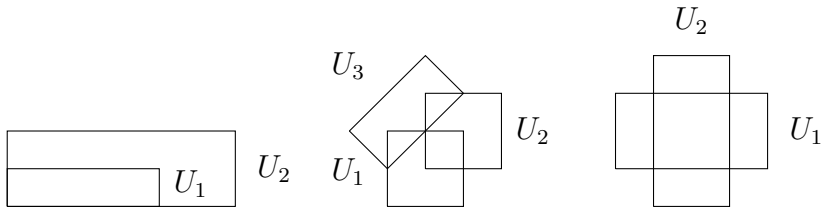
\begin{figure}[H]
\centering
\begin{tikzpicture}
\draw[-] (-1,0)--(1,0)--(1,.5)--(-1,.5)--(-1,0) {};
\draw[-] (-1,0)--(2,0)--(2,1)--(-1,1)--(-1,0) {};
\node[label= right:$U_1$] (U1) at (.9,.3) {};
\node[label=right:$U_2$] (U2) at (2,.5) {};

\draw[-] (4,0)--(5,0)--(5,1)--(4,1)--(4,0) {};
\draw[-] (4.5,.5)--(5.5,0.5)--(5.5,1.5)--(4.5,1.5)--(4.5,0.5) {};
\draw[-] (3.5,1)--(4,.5)--(5,1.5)--(4.5,2)--(3.5,1) {};
\node[label=above left:$U_1$] (U1) at (4,0) {};
\node[label=right:$U_2$] (U2) at (5.4,1) {};
\node[label=above left: $U_3$] (U2) at (4,1.4) {};

\draw[-] (7,.5)--(9,.5)--(9,1.5)--(7,1.5)--(7,.5) {};
\draw[-] (7.5,0)--(7.5,2)--(8.5,2)--(8.5,0)--(7.5,0) {};
\node[label=right:$U_1$] (U1) at (9,1) {};
\node[label=above: $U_2$] (U2) at (8,2) {};
\end{tikzpicture}
\caption{Realizations that violate the three parts of the definition of well-formed.}\label{wellformedviolationsfig}
\end{figure}

\begin{ex}
Consider the code $\C=\{000,100,001,101,011,111\}$.  This neural code is convex open and realizable in dimension 2.  A well-formed realization of $\C$ is shown in Figure~\ref{realization}.  The realization comprises 6 zones, one for each codeword, including the empty zone.
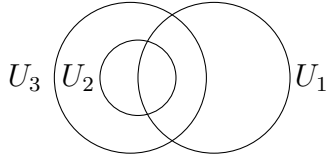
\begin{figure}[H]
\centering
\begin{tikzpicture}
\draw (-0.6,0) circle (1cm) {};
\draw (0.5,0) circle (1cm) {};
\draw (-0.5,0) circle (0.5cm) {};
\node[label=right :$U_3$] (3) at (-2.5,0) {};
\node[label=left : $U_1$] (1) at (2.3,0) {};
\node[label=right :$U_2$] (2) at (-1.8,0) {};
\end{tikzpicture}
\caption{A well-formed realization of the code $\mathcal{C}=\{000,100,001,101,011,111\}$.}\label{realization}
\end{figure}
\end{ex}


\begin{defn}
Let $\C=\C(\mathcal{U})$ be a well-formed neural code.  A place field $U_\ell$ is a \textit{$k$-piercing} of $U_1,\ldots, U_k$ \textit{identified by a zone} $z$ if the boundary of $U_\ell$ intersects the boundaries of $k$ other place fields $U_1,\ldots, U_k$ and the presence of $U_\ell$ adds exactly $2^k$ zones when $U_\ell$ is added to the realization $\mathcal{U}\setminus U_{\ell}$.  A $k$-piercing is identified by a zone in the diagram in which the place field $U_\ell$ is contained.  Note that the identifying zone $z$ is not necessarily unique.
\end{defn}

For example, in Figure~\ref{realization}, $U_2$ is a 1-piercing of $U_1$ identified by the zone $001$.  Notice that $U_1$ is not a 2-piercing; its boundary intersects the boundaries of two place fields, but adding it does not add $4$ new regions to the diagram.

\begin{defn}
Let $\C=\{\sigma_1,\ldots,\sigma_m\}$ be a code on $n$ neurons. We say $\C$ is {\em $k$-inductively pierced} if there exists a $k$-piercing $U_\lambda$ in $\C$ such that  $\C \setminus \{\lambda\} := \{\hat\sigma_1,\ldots ,\hat\sigma_m\}$ is $k$-inductively pierced where $\hat\sigma_i = \sigma_i$ except $\hat\sigma_i$ has a 0 in the $\lambda$th position. 

In this paper we will assume that if a code is $k$-inductively pierced, the code is not $(k+1)$-inductively pierced.
\end{defn}

Notice that the definition of a $k$-inductively pierced code relies on knowing a realization for $\C=\C(\mathcal{U})$.  Our main goal is to determine $k$-inductively pierced codes directly from the code, without knowledge of the arrangement of the $U_i\in \mathcal{U}$.  To understand $k$-inductively pierced codes directly from the code, the authors in \cite{GOY16} defined the toric ideal of a neural code.  

\begin{defn}
For a neural code $\C$, define the ring homomorphism 
\begin{align*}
\phi_\mathit{c}:\mathbb{F}_2[p_c\ |\ c\in\mathcal{C}]&\longrightarrow \mathbb{F}_2[x_i\ |\ i\in [n]] \\ 
p_c&\mapsto \prod_{i\in \supp(c)}x_i.
\end{align*}  Then, $I_\C:=\ker \phi_\C$ is the \textit{toric ideal} of $\C$. 
\end{defn} 

The toric ideal has the computational benefit of being relatively quick to compute using $\tt{Macaulay2}$ \cite{M2} with the $\tt{4ti2}$ package \cite{4ti2}, or with $\tt{SAGE}$ \cite{sage}.

In \cite{GOY16}, the authors investigated algebraic signatures for $0$-, $1$-, and $2$-inductively pierced codes by considering generators of $I_\C$.  They gave necessary and sufficient conditions on the toric ideals for $0$-inductively pierced codes, as well as a necessary condition on the toric ideal for $1$-inductively pierced codes.

\begin{prop}[\cite{GOY16}]\label{exisitingresults}
Let $\mathcal{C}$ be a well-formed neural code on $n$ neurons.
\begin{enumerate}
\item The neural code $\mathcal{C}$ is 0-inductively pierced if and only if $\mathit{I}_{\mathcal{C}}=\langle 0 \rangle.$ 
\item If the neural code $\C$ is 0- or 1- inductively pierced, then $\mathit{I}_{\mathcal{C}}=\langle 0 \rangle$ or is generated by quadratics.
\item If the neural code $\C$ contains a triple intersection of three place fields and the pairwise intersections are general, then $\mathit{I}_{\mathcal{C}}$ contains a binomial of degree 3 of particular form, in particular $p_{111w}p^2_{000v}-p_{100v}p_{010v}p_{001w}$ or $p_{111w}-p_{100\ldots 0}p_{010\ldots 0}p_{001w}$ where $v,w\in \{0,1\}^{n-3}$ correspond to zones in $\mathcal{C}(\mathcal{U})$.
\end{enumerate}
\end{prop}

The authors showed that the converse of the part 2 of Proposition~\ref{exisitingresults} is false; in fact, there exists a counter-example on as few as three neurons.

\begin{ex}
Let A1=$\{000,100,010,001,110,101,011,111\}$.  A generating set for $I_{\text{A1}}$, is $\{p_{110}-p_{100}p_{010},\ p_{101}-p_{100}p_{001},\ p_{011}-p_{010}p_{001},\ p_{111}-p_{110}p_{001}\}$. The toric ideal $I_{\text{A1}}$ is generated by quadratics, but as seen in Figure~\ref{A1fig}, each of $U_1,U_2,U_3$ are 2-piercings of the other two place fields, and so A1 is not 1-inductively pierced.  In fact, A1 is 2-inductively pierced.

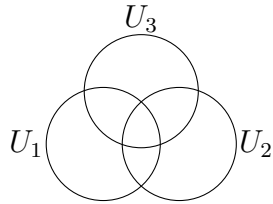
\begin{figure}[h!]
\begin{center}
\begin{tikzpicture}
\draw (-0.5,0) circle (0.75cm);
\draw (0.5,0) circle (0.75cm);
\draw (0,0.7) circle (0.75cm);

\node[label=left :  $U_1$] (1) at (-1,0) {};
\node[label=right :  $U_2$] (2) at (1,0) {};
\node[label=above :  $U_3$] (3) at (0,1.2) {};
\end{tikzpicture}
\caption{A realization of A1, a 2-inductively pierced code whose toric ideal is generated by quadratics.}\label{A1fig}
\end{center}
\end{figure}
\end{ex}

We note that part 3 of Proposition~\ref{exisitingresults} includes the case of codes with 2-piercings (Proposition~4.3.1 in~\cite{obatake2016}) and the case of $2$-inductively pierced codes, as summarized by Lemma~\ref{2piercingnecc}. 

\begin{lem}\label{2piercingnecc}
If a neural code $\C$ contains a 2-piercing or is 2-inductively pierced, then $I_\C$ contains a cubic of form $p_{111w}p^2_{000v}-p_{100v}p_{010v}p_{001w}$ or $p_{111w}-p_{100\ldots 0}p_{010\ldots 0}p_{001w}$ where $v=w$.
\end{lem}

We refer to the binomial $p_{111w}p^2_{000v}-p_{100v}p_{010v}p_{001w}$ for $w,v\in \{0,1\}^{n-3}$ as a \textit{cubic of a particular form} or a \textit{particular cubic}.  

In the next section, we investigate sufficient conditions for $1$- and $2$-inductively pierced codes, and strengthen parts 2 and 3 of Proposition~\ref{exisitingresults}.

\section{Main Results}



Throughout the rest of this paper, we assume that $\C=\C(\mathcal{U})$ is a convex open well-formed neural code on $n$ neurons that is realizable in $\R^2$, where $\mathcal{U}=\{U_1,\ldots,U_n\}$ and  $U_i\subset \R^2$, and $I_\C$ is the toric ideal of $\C$. 

Our main results use elements of the toric ideal as signatures for 1- and 2-inductively pierced codes.  Proposition~\ref{1-inductively_pierced_sufficient_condition} gives a sufficient condition for 1-inductively pierced codes, and Theorem~\ref{w=vimp2-ind} gives a sufficient condition for 2-inductively pierced codes.

We will often need to restrict our attention to certain neurons and their corresponding place fields.  For this, we define the restricted code.

\begin{defn}
Let $\C$ be a well-formed code on $n$ neurons, say $\C=\{\sigma_1,\ldots,\sigma_m\}$.  For $\Lambda\subseteq [n]$, define the {\em restricted code} $\C|_{\Lambda}:=\{\hat{\sigma}_1,\ldots,\hat{\sigma}_m\}$ on $n$ neurons, where $\hat{\sigma}_i$ has $\hat{\sigma}_{i_j}=\sigma_{i_j}$ for $j\in\Lambda$ and $\hat{\sigma}_{i_k}=0$ for $k\not\in\Lambda$ (here $\sigma_{i_\ell}$ denotes the $\ell$th component of the codeword $\sigma_i\in \C$).  
\end{defn}

\begin{ex}
Consider the code $\C = \{000,100,001,101,011,111\}$ from Example~\ref{realization}.  If we want to restrict our attention to the firing activity of neurons 1 and 2, we take $\Lambda=\{1,2\}$, so then $\C|_\Lambda=\{000,100,000,100,010,110\}=\{000,100,010,110\}$.  
\end{ex}

\begin{prop}\label{cubic_implies_2-piercing}
If there exists a cubic element of $I_\C$ of the form $p_{111w}p_{000v}^2-p_{100w}p_{010v}p_{001v}$ (where $w,v\in \{0,1\}^{n-3}$), then $\C|_{\{1,2,3\}}$ is 2-inductively pierced.
\end{prop}

\begin{proof}
The existence of such a cubic in $I_\C$ implies that the code contains the triple intersection of $U_1, U_2, \text{and\ } U_3$, so $U_1\cap U_2\cap U_3 \neq \emptyset$, along with the associated singleton zones, so $U_i\setminus (U_j\cup U_k)\neq \emptyset$ for $i,j,k\in \{1,2,3\}$, with $i\neq j \neq k$. Since the code is well formed, it must also contain the pairwise intersections, so $U_i\cap U_j\neq \emptyset$ for $i,j\in\{1,2,3\}$. Thus, if all other neurons are removed, the remaining code is clearly 2-inductively pierced, since $\C|_{\{1,2,3\}}=$A1.
\end{proof}

\begin{cor}\label{1-inductively_pierced_violation}
If $p_{111w}p_{000v}^2-p_{100w}p_{010v}p_{001v} \in I_\C$, then $\C$ is not 1-inductively pierced.
\end{cor}

\begin{proof}
Assume $\C$ is 1-inductively pierced. By Proposition~\ref{cubic_implies_2-piercing}, $\C|_{\{1,2,3\}}$ contains a 2-piercing. Thus, $U_1$ cannot be a 0- or 1-piercing unless either $U_2$ or $U_3$ is removed. The same is true of $U_2$ and $U_3$, so $U_1$, $U_2$, and $U_3$ cannot be 0- or 1-piercings. As place fields that are 0- or 1-piercings are removed, eventually $\C|_{\{1,2,3\}}$ must be all that is left, but $\C|_{\{1,2,3\}}$ has a 2-piercing, thus $\C$ is not 1-inductively pierced, resulting in a contradiction.
\end{proof}


\begin{lem}\label{1-piercingsuff}
Let $\C$ be a well-formed code on $n$ neurons such that in $\C|_{\{1,2\}}$, $U_1$ is a 1-piercing of $U_2$ and there do not exist $i,j,k\in [n]$ such that in $\C|_{\{i,j,k\}}$, $U_i$ is a 2-piercing of $U_j$ and $U_k$.  Then $\C$ is 1-inductively pierced. 
\end{lem}

\begin{proof}
Suppose towards contradiction that $\C$ is not 1-inductively pierced. Then, there must be some neuron $s$ such that $U_s$ intersects both $U_1$ and $U_2$ such that fewer than 4 zones are created and neither $U_1$ nor $U_2$ is a 2-piercing of $U_s$. However, this can only happen if $\C$ is not well-formed. But by hypothesis, $\C$ is well-formed, resulting in a contradiction.
\end{proof}

\begin{thm}\label{w=vimp2-ind}
Let $\C$ be a well-formed code on $n$ neurons, and let $I_\C$ be its toric ideal. If there exists a cubic of the form $p_{111w}p_{000v}^2-p_{100w}p_{010v}p_{001v} \in I_\C$ and in all such cases $w=v$, then $\C$ is 2-inductively pierced. 
\end{thm}

\begin{proof}

We will prove the contrapositive.  That is, we will assume that $\C$ is not 2-inductively pierced and show that whenever the toric ideal contains cubics of the particular form, at least one such cubic satisfies $w\neq v$. 


By Proposition~\ref{cubic_implies_2-piercing}, $p_{111w}p_{000v}^2-p_{100w}p_{010v}p_{001v} \in I_\C$ implies that $\C|_{\{1,2,3\}}$ is 2-inductively pierced, so each of $U_1$, $U_2$, $U_3$ is a two-piercing of the other two.  

Suppose that there exists some $U_j$ that obstructs $U_1$, $U_2$, and $U_3$ from being 2-piercings of each other in $\C$.  If $U_j$ contains the triple intersection $U_1\cap U_2\cap U_3$, then all zones other than the triple intersection are partially contained in $U_j$ since $\C$ is well-formed.  If $U_j$ does not contain the triple intersection, then $U_j$ partially contains it.  In either case, we have the zones necessary to construct a cubic with $v\neq w$, as illustrated in Figure~\ref{proof_example}.  

In particular, as in the case of the left figure of Figure~\ref{proof_example}, when $U_4$ contains the entire intersection $U_1\cap U_2\cap U_3$, we see the zones corresponding to the codewords 0000, 0100, 0010, 1001, 1111.  So, we can construct the particular cubic $p_{1111}p^2_{0000}-p_{1001}p_{0100}p_{0010}$ which is in $I_\C$ and is such that $w=1\neq 0=v$ (here $w,v$ are the fourth components of the codewords).  Generalizing this, we can form a particular cubic with $w\neq v$ whenever some $U_j$ contains a triple intersection $U_1\cap U_2\cap U_3$ via the codewords corresponding to the intersection of all 4 zones, a region outside all of $U_1,U_2,U_3$, the zone corresponding to the pairwise intersection of $U_j$ and $U_1$ and the other singleton zones outside of $U_1,U_j$ corresponding to the regions within $U_2$ and $U_3$.  Similarly, in the case of the right figure of Figure~\ref{proof_example}, when $U_4$ partially contains the intersection $U_1\cap U_2\cap U_3$, we see the zones corresponding to the codewords 0000, 0001, 0100, 1010, 1111, and so we can construct the particular cubic $p_{1111}p^2_{0000}-p_{1010}p_{0100}p_{0001}$, where here $w=1\neq 0=v$ are the third components of the codewords.  Essentially this case is identical to the first case; the roles of $U_3$ and $U_4$ are switched.  As such we can still obtain the necessary zones to obtain a particular cubic via the generalization as above. 

Furthermore, using the ideas in the proof of Corollary~\ref{1-inductively_pierced_violation}, the only way for $\C$ to fail to be 2-inductively pierced is if there exists such a $U_j$, that is if there is a $U_j$ such that $U_1\cap U_2\cap U_3\subseteq U_j$ or $U_j\cap \left(U_1\cap U_2\cap U_3\right)\neq \emptyset$.  That is to say, there must be a curve $U_j$ that obstructs $U_1$, $U_2$, and $U_3$ from all being two piercings of each other, so $U_j$ cannot be a 0- or 1-piercing. 

Thus, there must exist a cubic of form $p_{111w}p_{000v}^2-p_{100w}p_{010v}p_{001v} \in I_\C$ where $v\neq w$.
\end{proof}

\begin{prop}\label{1-inductively_pierced_sufficient_condition}
Let $\C$ be a well-formed code on $n$ neurons. If $I_\C\neq \langle 0 \rangle$ and there is no cubic of the form $p_{111w}p_{000v}^2-p_{100w}p_{010v}p_{001v} \in I_\C$, then $\C$ is 1-inductively pierced. 
\end{prop}

\begin{proof}
The lack of a cubic of the form $p_{111w}p_{000v}^2-p_{100w}p_{010v}p_{001v} \in I_\C$ implies that there do not exist $i,j,k\in [n]$ such that in $\C|_{\{i,j,k\}}$, $U_i$ is a 2-piercing of $U_j$ and $U_k$. Since $I_\C \neq \langle 0 \rangle$, $\C$ cannot be 0-inductively pierced by Proposition~\ref{exisitingresults}, thus there must be some curves that intersect generally, i.e. there are some $U_\ell$ and $U_m$ such that $U_\ell\cap U_m\neq \emptyset$, where $\partial U_\ell\cap \partial U_m\neq \emptyset$. Thus by Lemma~\ref{1-inductively_pierced_violation}, $\C$ must be 1-inductively pierced.
\end{proof}

\begin{figure}
\begin{center}
\begin{tikzpicture}[scale=0.7]
\draw (-1.4,0) circle (2cm);
\draw (1.4,0) circle (2cm);
\draw (0,2) circle (2cm);
\draw (0,.8) circle (1.75cm);

\node[label=left :  $U_1$] (1) at (-3.3,0) {};
\node[label=right :  $U_2$] (2) at (3.25,0) {};
\node[label=above :  $U_3$] (3) at (0,3.7) {};
\node[label=above right: $U_4$] (4) at (0,2.2) {};

\node[label=above: {\scriptsize $0010$}] (5) at (-0.5, 2.8) {};
\node[label=above: {\scriptsize $0000$}] (6) at (-2.8,2) {};
\node[label=above: {\scriptsize $1111$}] (7) at (0,0) {};
\node[label=above: {\scriptsize $0100$}] (8) at (2.1,-1.7) {};
\node[label=above: {\scriptsize $1001$}] (9) at (-1,-0.7) {};
\end{tikzpicture}
\qquad 
\begin{tikzpicture}[scale=0.7]
\draw (-1.4,0) circle (2cm);
\draw (1.4,0) circle (2cm);
\draw (0,2) circle (2cm);
\draw (0,.8) circle (1.75cm);

\node[label=left :  $U_1$] (1) at (-3.3,0) {};
\node[label=right :  $U_2$] (2) at (3.25,0) {};
\node[label=above :  $U_4$] (3) at (0,3.7) {};
\node[label=above right: $U_3$] (4) at (0,2.2) {};

\node[label=above: {\scriptsize $0001$}] (5) at (-0.5, 2.8) {};
\node[label=above: {\scriptsize $0000$}] (6) at (-2.8,2) {};
\node[label=above: {\scriptsize $1111$}] (7) at (0,0) {};
\node[label=above: {\scriptsize $0100$}] (8) at (2.1,-1.7) {};
\node[label=above: {\scriptsize $1010$}] (9) at (-1,-0.7) {};
\end{tikzpicture}
\end{center}
\caption{Examples of realizations of codes that are not 2-inductively pierced and have cubics of a particular form where $v\neq w$ in $I_\C$. On the left, $U_4$ contains the entire triple intersection $U_1\cap U_2\cap U_3$, and on the right $U_4$ contains part of the triple intersection.}
\label{proof_example}
\end{figure}
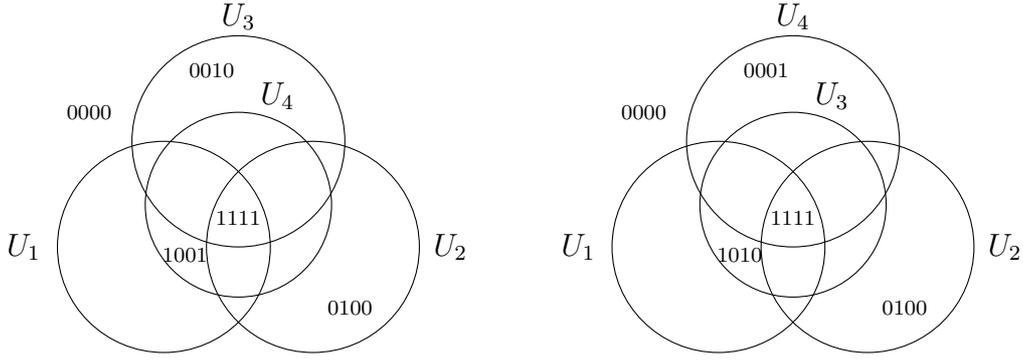








The existence of particular cubic binomials in the toric ideal can tell us much about the code. However, computing the entire toric ideal is computationally expensive; we would much rather only look at generating sets. Unfortunately, the existence of particular cubics in the toric ideal does not imply they are generators. 

\begin{ex}\label{A1}
As given in \cite{obatake2016}, one generating set for the toric ideal of the code A1 = \{000, 100, 010, 001, 110, 101, 011, 111\} is $\{p_{111}-p_{100}p_{010}p_{001},\ p_{110}-p_{100}p_{010},\ p_{101}-p_{100}p_{001},\ p_{011}-p_{010}p_{001}\}$.  This generating set has a particular cubic.  Another generating set for for $I_{\text{A1}}$, which consists only of quadratics is $\{p_{110}-p_{100}p_{010},\ p_{101}-p_{100}p_{001},\ p_{011}-p_{010}p_{001},\ p_{111}-p_{110}p_{001}\}$.
\end{ex}

From the second generating set for $I_{A1}$ in Example~\ref{A1}, we notice that \textit{a priori} we may not see the existence of a particular cubic in the toric ideal of a code from the generators directly.  However, we show that there is a class of quadratic generators that can imply the existence of the particular cubics in the toric ideal.

\begin{defn} \label{quadofpartform}
Let $\C$ be code on $n$ neurons and $w,v\in \{0,1\}^{n-3}$. A pair of quadratics having one of the following forms:
\begin{align*}
\begin{cases} p_{111w}p_{000v}-p_{110w}p_{001v}\\
p_{110w}p_{000v}-p_{100w}p_{010v}\\
\end{cases}
\mathrm{or\ }
\begin{cases} p_{111w}p_{000v}-p_{101w}p_{010v}\\
p_{101w}p_{000v}-p_{100w}p_{001v}\\
\end{cases}
\mathrm{or\ }
\begin{cases} p_{111w}p_{000v}-p_{011w}p_{100v}\\
p_{011w}p_{000v}-p_{010w}p_{001v}
\end{cases}
\end{align*}
is called \textit{a friendly quadratic pair}.
\end{defn}

\begin{prop}\label{friendlyquadraticimpcubic}
Let $w,v\in \{0,1\}^{n-3}$.  If any of the friendly quadratic pairs is in the toric ideal $I_\C$ of $\C$, then $I_\C$ contains a cubic of the form $p_{111w}p_{000v}^2-p_{100w}p_{010v}p_{001v}$, and thus by Corollary~\ref{1-inductively_pierced_violation}, $\C$ is not 1-inductively pierced.


\end{prop}

\begin{proof}
Consider the first case, so suppose $I_\C$ contains the first friendly quadratic pair $p_{111w}p_{000v}-p_{110w}p_{001v}, p_{110w}p_{000v}-p_{100w}p_{010v}$.  Recall that $I_\C$ is an ideal, hence is closed under multiplication by elements of the ring.  Since $p_{111w}p_{000v}-p_{110z}p_{001z}\in I_\C$, we get that $p_{000w},p_{001v}\in \mathbb{F}_2[p_c\ |\ c\in \C]$ and so, 
\begin{align*}
p_{111w}p_{000v}^2-p_{100w}p_{010v}p_{001v} =p_{000v}&(p_{111w}p_{000v}-p_{110w}p_{001v})\\&+p_{001v}(p_{110w}p_{000v}-p_{100w}p_{010v})\in I_\C.
\end{align*}  The proof in the other two cases is similar.
\end{proof}


Our main results rely on identifying particular cubics in the toric ideal of a code.  The main advantage of Proposition~\ref{friendlyquadraticimpcubic} is that it allows us to use friendly quadratic pairs that are generators of the toric ideal to quickly determine the existence of particular cubics in the toric ideal, even when they do not explicitly appear in a generating set of $I_\C$.  We illustrate this with an example.

\begin{ex}\label{friendlyexample}
Consider the following neural code on 5 neurons, $\C
=$\{00001, 10001, 01001, 00011, 11001, 10011, 01011, 00111, 11011, 10111, 01111, 11111\}. One generating set for $I_\C$ is 
\begin{align*}
\langle p_{00111}p_{11111} - p_{10111}p_{01111},\ p_{01011}p_{11111} -
p_{11011}p_{01111},\ p_{10011}p_{11111} -p_{11011}p_{10111},\\
p_{00011}p_{11011} -p_{10011}p_{01011},\ p_{00011}p_{10111} -p_{10011}p_{00111},\ p_{00011}p_{11111} -p_{10011}p_{01111},\\ 
p_{00011}p_{01111} -p_{01011}p_{00111},\ p_{00011}p_{11111} -p_{01011}p_{10111},\ p_{00011}p_{11111} -p_{00111}p_{11011},\\
p_{01001}p_{10011} -p_{00011}p_{11001},\ p_{01001}p_{11011} - p_{11001}p_{01011},\ p_{01001}p_{10111} - p_{11001}p_{00111},\\
p_{01001}p_{11111} - p_{11001}p_{01111},\ p_{10001}p_{01011} - p_{00011}p_{11001},\ p_{10001}p_{11011} - p_{11001}p_{10011},\\ 
p_{10001}p_{01111} - p_{11001}p_{00111},\ p_{10001}p_{11111} - p_{11001}p_{10111},\ \textcolor{blue}{p_{00001}p_{11001} - p_{10001}p_{01001}},\\ 
p_{00001}p_{10011} - p_{10001}p_{00011},\ p_{00001}p_{10111} - p_{10001}p_{00111},\ p_{00001}p_{01011} - p_{01001}p_{00011},\\ 
	p_{00001}p_{01111} - p_{01001}p_{00111},\ p_{00001}p_{11011} - p_{00011}p_{11001},\ \textcolor{blue}{p_{00001}p_{11111} - p_{11001}p_{00111}}\rangle.
\end{align*}

Note that no cubics appear in this generating set. However, $p_{00001}p_{11111} - p_{11001}p_{00111}$ and $p_{00001}p_{11001} - p_{10001}p_{01001}$ do. This is a friendly quadratic pair.  Thus a cubic of particular form must appear in $I_\C$, namely a particular cubic with $w=v$. Thus, by Proposition~\ref{friendlyquadraticimpcubic}, $I_\C$ contains a particular cubic, and so by Corollary~\ref{1-inductively_pierced_violation}, $\C$ is not 1-inductively pierced.

We see from a realization of $\C$ in Figure~\ref{examplecode} that this code is 2-inductively pierced, even though it does not satisfy the conditions of Theorem~\ref{w=vimp2-ind}.  So Theorem~\ref{w=vimp2-ind} gives a sufficient condition for a code to be 2-inductively pierced, but it is not a necessary condition.

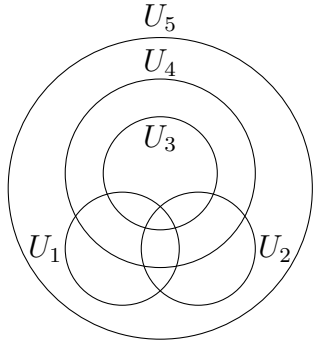
\begin{figure}[h!]
\begin{center}
\begin{tikzpicture}
\draw (-0.5,0) circle (0.75cm);
\draw (0.5,0) circle (0.75cm);
\draw (0,1) circle (0.75cm);
\draw (0,1) circle (1.25cm);
\draw (0,0.8) circle (2cm);

\node[label=left :  $U_1$] (1) at (-1,0) {};
\node[label=right :  $U_2$] (2) at (1,0) {};
\node[label=above :  $U_3$] (3) at (0,1) {};
\node[label=above: $U_4$] (4) at (0,2) {};
\node[label=above: $U_5$] (4) at (0,2.6) {};

\end{tikzpicture}
\end{center}
\caption{A realization of the code from Example~\ref{friendlyexample}. }\label{examplecode}
\end{figure}



\end{ex}

\section{Discussion}
We have shown sufficient conditions for both 1- and 2-inductively pierced codes. However, our sufficient condition for a code to be 1-inductively pierced relies on being able to analyze the entire toric ideal, a task that becomes unfeasible for codes on larger numbers of neurons. On the other hand, analyzing generating sets of the toric ideal is much easier. If we could classify all possible ways that a cubic of particular form can be generated, we could identify whether a cubic of particular form is in the toric ideal simply by looking at an arbitrary generating set. If this were accomplished, Theorem~\ref{w=vimp2-ind} and Proposition~\ref{1-inductively_pierced_sufficient_condition} would become much more powerful. 

Another direction for future study lies in identifying which receptive fields could potentially form a $2$-piercing. In this paper we assumed that we knew which three place fields were involved in a $2$-piercing just from looking at the code. One could simply check every possible combination of three neurons, but as the number of neurons increases, this becomes far more difficult. 

Further research in identifying which codes are well-formed and realizable in two dimensions would also be of great benefit. We have made the assumption that the codes we work with are well-formed and realizable in $\R^2$, but as of yet we are unaware of the existence of any tools to determine these conditions. 


\section*{Acknowledgements}
This research was conducted during the summer REU program at Texas A\&M University and funded by NSF DMS-1460766. We are grateful to Dr.~Anne Shiu for her valuable insight throughout the research process and her helpful comments on this manuscript.

\bibliographystyle{abbrv}
\bibliography{REU2017}

\end{document}